\UseRawInputEncoding
\documentclass[10pt]{amsart}
\usepackage{amsfonts}
\usepackage{amssymb}
\newtheorem{theorem}{Theorem}
\newtheorem{lemma}{Lemma}
\newtheorem{corollary}{Corollary}

\newtheorem{remark}{Remark}

\begin{document}

\title[Spectral Theorem approach to Quantum Observables II]{Spectral Theorem approach
to the  Characteristic Function of Quantum Observables  II}
\author{Andreas Boukas}
\address{Centro Vito Volterra, Universit\`{a} di Roma Tor Vergata, via Columbia  2, 00133 Roma,
Italy and Graduate School of Mathematics, Hellenic Open
University, Greece}
\email{boukas.andreas@ac.eap.gr}

\date{\today}



\begin{abstract}
We compute the resolvent of the \textit{anti-commutator operator}
$XP+PX$ and of the \textit{quantum harmonic oscillator Hamiltonian
operator} $\frac{1}{2}(X^2+P^2)$. Using Stone's formula for
finding the spectral resolution of an, either bounded or
unbounded, self-adjoint operator on a Hilbert space,  we also
compute their Vacuum Characteristic Function (Quantum Fourier
Transform). We also show how Stone's formula is applied to the
computation of the Vacuum Characteristic Function of finite
dimensional quantum observables. The method is proposed as an
analytical alternative to the algebraic (or Heisenberg) approach
relying on the associated Lie algebra commutation relations.
\end{abstract}

\maketitle

\section{Introduction}\label{intro}

On page $373$ of \cite{taylor} it is stated that "the actual
determination of the resolution of the identity for a given
operator is not an easy matter, in general". Some Functional
Analysis texts that cover spectral integration, for example
\cite{Yosida}, give, at most, the spectral resolution of the
multiplication and differentiation operators on the Hilbert space
of square integrable functions. Few Functional Analysis texts,
such as \cite{DS}, give Stone's formula for computing the spectral
resolution of a self-adjoint Hilbert space operator and even
fewer, for example \cite{Ro}, give examples of how to use it.

In this paper, which is a sequel to \cite{BFspec}, we illustrate
in detail the use of Stone's formula for the spectral resolution,
by applying it to the  \textit{anti-commutator operator} $XP+PX$
and the \textit{quantum harmonic oscillator Hamiltonian operator}
$\frac{1}{2}(X^2+P^2)$ of Quantum Mechanics. We see that it
naturally leads to the appearance of several types of differential
equations and special functions. After computing the spectral
resolution, we compute the Vacuum Characteristic Function (Quantum
Fourier Transform) of these operators and show that the results
agree with those obtained, using Lie algebraic techniques, in
\cite{AccBouCOSA}.

 As in \cite{BFspec}, we consider the (self-adjoint) \textit{position, momentum} and
 \textit{identity} operators, respectively, defined in $L^2(\mathbb{R},\mathbb{C})$ with inner
product
\[
\langle f, g \rangle
=\frac{1}{\sqrt{\hslash}}\,\int_{\mathbb{R}}\,\overline{f(x)}g(x)\,dx
\ ,
\]
by
\[
X\,f(x)=x\,f(x)\,\,;\,\,
P\,f(x)=-i\,\hslash\,f^{\prime}(x)\,\,;\,\, \mathbf{1}\,f(x)=I
f(x)=f(x)\    ,
\]
and satisfying the commutation relations
\[
\lbrack P, X \rbrack=-i\hslash\mathbf{1}  \ ,
\]
on
\[
\Omega={\rm dom }(X)\cap {\rm dom }(P) \  ,
\]
 where,
\[
{\rm dom }(X)=\{f\in L^2(\mathbb{R},\mathbb{C}))
\,:\,\int_{\mathbb{R}}x^2\,|f(x)|^2\,dx<+\infty\} \ ,
\]
and
\[
{\rm dom }(P)=\{f\in L^2(\mathbb{R},\mathbb{C}) \,:\,\mbox{$f$ is
absolutely continuous and
}\int_{\mathbb{R}}\,\left|f^\prime(x)\right|^2\,dx< +\infty\} \ ,
\]
are respectively the, dense in $L^2(\mathbb{R},\mathbb{C})$,
domains of $X$ and $P$ . Since it contains
$C_0^{\infty}(\mathbb{R})$, $\Omega$ is nonempty and dense in
$L^2(\mathbb{R},\mathbb{C})$. The \textit{Schwartz class}
$\mathcal{S}$ is a common invariant domain \cite{FG} of $X$ and
$P$ which is also dense in $L^2(\mathbb{R},\mathbb{C})$ and
contains $C_0^{\infty}(\mathbb{R})$, and  is therefore suitable
for defining $XP$ and $PX$. In this paper we normalize to
$\hslash=1$. Functions in the domain of $P$ are continuous and
vanish at infinity. As pointed out in \cite{BFspec},
\[
\Phi=\Phi(x)=
\pi^{-1/4}\,e^{-\frac{x^2}{2\hslash}}=\pi^{-1/4}\,e^{-\frac{x^2}{2}}\
,
\]
is a unit vector in $\Omega$. For $a\in\mathbb{C}$ we denote
$R(a;T)=(a-T)^{-1}$ the \textit{resolvent} of $T$. The
\textit{spectral resolution} $\{E_\lambda \,|\,\lambda \in
\mathbb{R} \}$ of a bounded or unbounded self-adjoint  operator
$T$ in a complex separable Hilbert Space $\mathcal{H}$ can be
computed using Stone's formula (see \cite{DS}, Theorems X.6.1 and
XII.2.10):
\[
E\left((a,b)\right)=\lim_{\delta\to 0^+}\lim_{\epsilon\to
0^+}\frac{1}{2\pi i}\int_{a+\delta}^{b-\delta}\left(R(t-\epsilon
i; T)-R(t+\epsilon i; T)\right)\,dt \  ,
\]
where $(a,b)$ is the open interval $a<\lambda<b$,  and the limit
is in the strong operator topology. For $a\to-\infty$ and
$b=\lambda$ we have
\begin{align*}
E_\lambda&=E\left( (-\infty, \lambda ]\right)=\lim_{\rho\to
0^+}E\left( (-\infty,
\lambda+\rho)\right)\\
&=\lim_{\rho\to 0^+}\lim_{\delta\to 0^+}\lim_{\epsilon\to
0^+}\frac{1}{2\pi
i}\int_{-\infty}^{\lambda+\rho-\delta}\left(R(t-\epsilon i;
T)-R(t+\epsilon i; T)\right)\,dt \   .
\end{align*}
For $f,g\in\mathcal{H}$, see \cite{Ro},
\begin{equation}\label{sr}
\langle f, E_\lambda g\rangle=\lim_{\epsilon\to 0^+}\frac{1}{2\pi
i}\int_{-\infty}^{\lambda}\langle f, \left(R(t-\epsilon i;
T)-R(t+\epsilon i; T)\right)g\rangle\,dt \   .
\end{equation}
Once the \textit{vacuum resolution of the identity}
\[
\langle \Phi, E_\lambda \Phi\rangle=\lim_{\epsilon\to
0^+}\frac{1}{2\pi i}\int_{-\infty}^{\lambda}\langle \Phi,
\left(R(t-\epsilon i; T)-R(t+\epsilon i; T)\right)\Phi\rangle\,dt
\ ,
\]
corresponding to the operator $T$ is known, equivalently, once its
\textit{vacuum differential}
\[
d\langle \Phi, E_\lambda \Phi\rangle \ ,
\]
is known, we can immediately compute its \textit{vacuum
characteristic function}
\[
\langle \Phi, e^{itT}\Phi
\rangle=\int_{\mathbb{R}}e^{it\lambda}\,d\langle\Phi,
E_{\lambda}\Phi \rangle\ .
\]
For a complex number $z$ we denote ${\rm Re }z$, ${\rm Im }z$ its
real and imaginary part, respectively.

The following Lemma is very useful in computing the vacuum
spectral resolution.

\begin{lemma}\label{fgs}
\[
\langle \Phi, E_\lambda \Phi\rangle=  \lim_{\epsilon\to
0^+}\frac{1}{\pi }\int_{-\infty}^{\lambda}{\rm Im }\langle
\Phi,R(t-\epsilon i; T)\Phi\rangle\,dt  \  .
\]
\end{lemma}

\begin{proof}Using the resolvent identity, see \cite{DS},
\[
R(\lambda; T)^*=R(\bar\lambda; T^*)\  ,
\]
 we have
\begin{align*}
\langle \Phi, E_\lambda \Phi\rangle=&\lim_{\epsilon\to
0^+}\frac{1}{2\pi i}\int_{-\infty}^{\lambda}\langle \Phi,
\left(R(t-\epsilon i; T)-R(t+\epsilon i; T)\right)\Phi\rangle\,dt\\
=&\lim_{\epsilon\to 0^+}\frac{1}{2\pi
i}\int_{-\infty}^{\lambda}\left( \langle  \Phi, R(t-\epsilon i;
T)\Phi\rangle-\langle R(t-\epsilon i;
T)\Phi,\Phi\rangle\right)\,dt\\
=&\lim_{\epsilon\to 0^+}\frac{1}{2\pi
i}\int_{-\infty}^{\lambda}\left( \langle  \Phi,
R(t-\epsilon i; T)\Phi\rangle-\overline{\langle \Phi,R(t-\epsilon i;
T)\Phi\rangle}\right)\,dt\\
=&\lim_{\epsilon\to 0^+}\frac{1}{2\pi i}\int_{-\infty}^{\lambda}2
i \,{\rm Im } \langle \Phi,R(t-\epsilon i;
T)\Phi\rangle\,dt \\
=&\lim_{\epsilon\to 0^+}\frac{1}{\pi }\int_{-\infty}^{\lambda}
{\rm Im } \langle \Phi,R(t-\epsilon i; T)\Phi\rangle\,dt   \  .
\end{align*}
\end{proof}

As in \cite{BFspec}, we define the Fourier transform of $f$ by
\[
\hat{f}(t)=\left(Uf\right)(t)=(2
\pi)^{-1/2}\,\int_{-\infty}^\infty \,e^{i\lambda
t}f(\lambda)d\lambda\ ,
\]
and the inverse Fourier transform of $\hat{f}$ by
\[
f(\lambda)=\left(U^{-1}\hat{f}\right)(\lambda)=(2
\pi)^{-1/2}\,\int_{-\infty}^\infty \,e^{-i\lambda t}\hat{f}(t)dt\
.
\]
Finally, for $a\in\mathbb{R}$ we denote by $\delta_a$ and $H_a$ ,
respectively, the \textit{Dirac delta} and \textit{Heaviside unit
step} functions defined, for a \textit{test function $f$}, by
\[
\int_{\mathbb{R}} f(x)\delta_a (x)\,dx =\int_{\mathbb{R}}
f(x)\delta (x-a)\,dx=f(a) \ ,
\]
 and
\[
H_a(x)=H(x-a)=\left\{
\begin{array}{llr}
 1\   ,&  x \geq a  \\
& \\
0\   ,& x < a
 \end{array}
 \right.  \  .
\]
In the sense of distributions, $\delta_a$  is the derivative of
$H_a$.

\section{The anti-commutator operator $XP+PX$}

\begin{theorem}\label{sa1} The operator $XP+PX$ is symmetric on the Schwartz class $\mathcal{S}$
 and admits a self-adjoint extension.
\end{theorem}

\begin{proof}  Let $T=XP+PX$ and let $f,g
\in \mathcal{S}$. Since $Xf, Pf \in \mathcal{S}$, $T$ is well
defined on ${\rm dom}(T)=\mathcal{S}$. Since $X, P$ are
self-adjoint,
\begin{align*}
\langle Tf, g\rangle=& \langle XPf, g \rangle+ \langle PXf, g
\rangle =\langle Pf, Xg \rangle+ \langle Xf, Pg \rangle\\
=&\langle f, PXg \rangle+ \langle f, XPg \rangle =\langle f,
(PX+XP)g \rangle =\langle f, Tg \rangle\ ,
\end{align*}
so $T$ is symmetric. To show that $T$ admits a self-adjoint
extension, we consider the \textit{ conjugation operator}
\[
Kf(x)=\overline{f(-x)} \  .
\]
If $f$ is in the domain of $T$ then so also is $Kf$ and, since for
a function $g$ in the domain of $T$,
\[
Tg(x)=-i g(x)-2 i x \frac{d}{dx}g(x) \  ,
\]
we find
\begin{align*}
\lbrack T, K \rbrack f (x)&=TKf(x)-KTf(x)\\
&=\left(-i \overline{f(-x)}+2 i x \overline{f'(-x)}
\right)-\left(-i \overline{f(-x)}+2 i x \overline{f'(-x)}
\right)=0\  ,
\end{align*}
 i.e., in the context of Section 8 of \cite{weidmann},
$T$ is $K$-real. Thus, by Theorem 8.9 of \cite{weidmann}, the
symmetric operator $T$ has a self-adjoint extension.
\end{proof}

\begin{theorem}\label{R1}For ${\rm Im} a>-1$ and $s\neq 0$, the resolvent of
 $XP+PX$ is
\[
R(a;XP+PX)g(s)=\left\{
\begin{array}{llr}
 \frac{i}{2}s^{ -\frac{a+i}{2i} }\int_s^\infty w^{ \frac{a-i}{2i}
}g(w)\,dw  \   ,&  s> 0   \\
& \\
\frac{i}{2}(-s)^{ -\frac{a+i}{2i} }\int_{-\infty}^s (-w)^{
\frac{a-i}{2i} }g(w)\,dw\ , & s<0
\end{array}
\right. \  ,
\]
where $g\in \mathcal{S}$. For $s=0$,
\[
R(a;XP+PX)g(0)=\frac{g(0)}{a+i}\  .
\]
\end{theorem}

\begin{proof} For $a\in\mathbb{C}$ with ${\rm Im} a>-1$ and for $s\in\mathbb{R}$,
\begin{align*}
R(a;XP+PX)g(s)=G(s)& \iff g(s)=(a-XP-PX)G(s)\notag\\
& \iff g(s)=a G(s)+i s G^\prime (s)+i (G(s))+s G^\prime (s))    \notag\\
& \iff 2 i s G^\prime(s)+(a+i) G(s)= g(s) \ .
\end{align*}
 For $s=0$ we find
\[
G(0)=\frac{g(0)}{a+i}\  .
\]
For $s\neq 0$, we have
\[
 G^\prime(s)+\frac{a+i}{ 2 i s} G(s)= \frac{g(s)}{ 2 i s} \ ,
\]
which is a first-order linear ordinary differential equation with
complex coefficients. Multiplying  by the integrating factor
$|s|^{ \frac{a+i}{2i} }$, replacing $s$ by $t$ and integrating
from $s$ to $\infty$  (if $s>0$) and from  $-\infty$
 to $s$ (if $s<0$), using the fact
that, since $G$ is in the Schwartz class,
\[
\lim_{t\to \pm \infty}|t^{ \frac{a+i}{2i} } G(t)|=\lim_{t\to \pm
\infty}|t^{ \frac{{\rm Im}a+1}{2} } G(t)| =0\   ,
\]
 we find
\[
G(s)=R(a; XP+PX)g(s)=\left\{
\begin{array}{llr}
 \frac{i}{2}s^{ -\frac{a+i}{2i} }\int_s^\infty w^{ \frac{a-i}{2i}
}g(w)\,dw  \   ,&  s> 0   \\
& \\
\frac{i}{2}(-s)^{ -\frac{a+i}{2i} }\int_{-\infty}^s (-w)^{
\frac{a-i}{2i} }g(w)\,dw\ , & s<0 \   .
\end{array}
\right.
\]
\end{proof}

\begin{corollary}\label{even} For ${\rm Im} a>-1$, if $g$ is even then so also is $R(a; XP+PX)g$.
\end{corollary}

\begin{proof} Let $s>0$. By Theorem \ref{R1},
\begin{align*}
R(a; XP+PX)g(-s)&=\frac{i}{2} s^{ -\frac{a+i}{2i}
}\int_{-\infty}^{-s} (-w)^{ \frac{a-i}{2i} }g(w)\,dw\\
&=\frac{i}{2}s^{ -\frac{a+i}{2i} }\int_{-\infty}^{-s} (-w)^{
\frac{a-i}{2i} }g(-w)\,(-1)d(-w) \  ,
\end{align*}
which, letting $u=-w$ in the integral, yields
\[
R(a; XP+PX)g(-s)=\frac{i}{2}s^{ -\frac{a+i}{2i} }\int_s^\infty u^{
\frac{a-i}{2i} }g(u)\,du=R(a; XP+PX)g(s) \ .
\]
For $s<0$,
\begin{align*}
R(a; XP+PX)g(-s)&=\frac{i}{2} (-s)^{ -\frac{a+i}{2i}
}\int_{-s}^{\infty} w^{ \frac{a-i}{2i} }g(w)\,dw\ ,
\end{align*}
which, letting $u=-w$, yields
\begin{align*}
R(a; XP+PX)g(-s)&= \frac{i}{2}(-s)^{ -\frac{a+i}{2i}
}\int_{s}^{-\infty} (-u)^{ \frac{a-i}{2i} }g(-u)\,(-1)du\\
&=\frac{i}{2}(-s)^{ -\frac{a+i}{2i}
}\int_{-\infty}^{s} (-u)^{ \frac{a-i}{2i} }g(u)\,du\\
&=R(a; XP+PX)g(s) \ .
\end{align*}
\end{proof}

\begin{theorem}\label{XP+PX} The vacuum spectral resolution of
 $XP+PX$ is
\[
\langle \Phi, E_\lambda
\Phi\rangle=\frac{1-i}{8\pi}\,\int_{-\infty}^\lambda
\left(e^{-\frac{\pi t}{4}}\,B\left(-1; \frac{1- i
t}{4},\frac{1}{2} \right)+e^{\frac{\pi t}{4}}\,B\left(-1; \frac{1+
i t}{4},\frac{1}{2} \right)\right)\,dt  \ ,
\]
where
\[
B(z; a, b)=\int_0^z t^{a-1} (1-t)^{b-1}\,dt=z^a \sum_{n=0}^\infty
\frac{(1-b)_n }{n!\,(a+n)}z^n\ ,
\]
is the \textit{incomplete Beta function} and, for
$x\in\mathbb{R}$,  $(x)_n=x(x+1)(x+2)\cdots(x+n-1)$. Moreover, for
$t\in \mathbb{R}$,
\[
\langle \Phi, e^{i t (XP+PX)} \Phi \rangle=\left({\rm sech} \,2t
\right)^{1/2}
\]
\end{theorem}

\begin{proof} We first present a proof without Stone's formula:
\[
\lbrack P, X\rbrack=-i\mathbf{1}\implies PX=XP-i\mathbf{1}\ ,
\]
so
\begin{align*}
\langle \Phi, e^{i t (XP+PX)} \Phi \rangle&=\langle \Phi, e^{i t
(2XP-i\mathbf{1})} \Phi \rangle=e^t \langle \Phi, e^{2i t XP} \Phi \rangle\\
&=e^t \sum_{n=0}^\infty \frac{(2it)^n}{n!}\langle \Phi, (XP)^n
\Phi \rangle \  .
\end{align*}
Using Lemma 8.4 of \cite{BFspec},  trivially extended to $n=0$ and
$\sum_{n=0}^\infty$ with the use of $S(0,0)=1$, we have
\begin{align*}
\langle \Phi, e^{i t (XP+PX)} \Phi \rangle=&e^t \sum_{n=0}^\infty
\frac{(2it)^n}{n!}\sum_{k=0}^n
(-1)^{n-k}i^{n-k} S(n,k) \langle \Phi, X^kP^k \Phi \rangle\\
=&e^t \sum_{n=0}^\infty \frac{(2it)^n}{n!}\sum_{k=0}^\infty
(-1)^{n-k}i^{n-k} S(n,k) \langle X^k\Phi, P^k \Phi \rangle\ ,
\end{align*}
where we have used the fact that $X$ is self-adjoint and the
\textit{Stirling numbers of the second kind} $S(n,k)$  satisfy
$S(n,k)=0$ for $k>n$. As in the proof of Theorem 8.5 of
\cite{BFspec},
\[
\langle X^k\Phi, P^k \Phi \rangle=\frac{1}{\pi
\sqrt{2}}\int_{\mathbb{R}}\int_{\mathbb{R}}\lambda^k \mu^k e^{
-\frac{\lambda^2+\mu^2}{2}+i\lambda \mu}\,d\lambda\,d\mu \  .
\]
Thus, switching the order of summation,
\begin{align*}
\langle \Phi, e^{i t (XP+PX)} \Phi \rangle=&\frac{e^t}{\pi
\sqrt{2}}\int_{\mathbb{R}}\int_{\mathbb{R}}\sum_{k=0}^\infty i^k
\left(\sum_{n=0}^\infty\frac{(2t)^n}{n!}S(n,k)\right)\lambda^k
\mu^k e^{ -\frac{\lambda^2+\mu^2}{2}+i\lambda \mu}\,d\lambda\,d\mu
\\
=&\frac{e^t}{\pi \sqrt{2}}\int_{\mathbb{R}}\int_{\mathbb{R}}e^{i
\lambda \mu (e^{2t}-1)}e^{ -\frac{\lambda^2+\mu^2}{2}+i\lambda
\mu}\,d\lambda\,d\mu\\
=&\frac{e^t}{\pi \sqrt{2}}\int_{\mathbb{R}}\int_{\mathbb{R}}e^{
-\frac{\lambda^2+\mu^2}{2}+i\lambda \mu e^{2t}}\,d\lambda\,d\mu\ ,
\end{align*}
where we have used the identities
\[
\sum_{n=0}^\infty\frac{(2t)^n}{n!}S(n,k)=\frac{(e^{2t}-1)^k }{k!}\
,
\]
and
\[
\sum_{k=0}^\infty\frac{ \left(  i  (e^{2t}-1) \lambda \mu
\right)^k  }{k!}=e^{i \lambda \mu (e^{2t}-1)} \  .
\]
Thus, using the integration formula
\[
\int_{\mathbb{R}}\int_{\mathbb{R}} e^{\alpha x^2+\beta y^2+i\gamma
x y}\, dx dy=\frac{2\pi}{\sqrt{\gamma^2+4 \alpha \beta}}\,\,,{\rm
Re}\, \beta<0\,\,,\,\, {\rm Re}\left(
4\alpha+\frac{\gamma^2}{\beta} \right)<0 \ ,
\]
we obtain
\[
\langle \Phi, e^{i t (XP+PX)} \Phi \rangle=\frac{e^t}{\pi
\sqrt{2}}\frac{2\pi}{\sqrt{e^{4t}+1}}=
\sqrt{\frac{2}{e^{2t}+e^{-2t}} }=\sqrt{\rm{sech}\, 2t}\ .
\]

The vacuum spectral resolution can be found using the inverse
Fourier transform (see Theorem 9.5 of \cite{BFspec}).

\medskip

 We can derive the formula for the vacuum spectral
resolution by a direct application of Stone's formula as follows:

\medskip

By Lemma \ref{fgs} and  Theorem \ref{R1},  since, by Corollary
\ref{even}, ${\rm Im }\left(\Phi(s)R(t-\epsilon i;
T)\Phi(s)\right)$ is an even function of $s$, we have
\begin{align*}
\langle \Phi, E_\lambda \Phi\rangle &=\frac{1}{\pi }
\lim_{\epsilon\to 0^+}\int_{-\infty}^{\lambda}{\rm Im }\langle
\Phi,R(t-\epsilon i;
T)\Phi\rangle\,dt\\
=&\frac{1}{\pi }  \lim_{\epsilon\to 0^+}\int_{-\infty}^{\lambda}
\int_{-\infty}^{\infty }{\rm Im }\left(\Phi(s)R(t-\epsilon i;
T)\Phi(s)\right)\,ds \,dt   \\
=&\frac{2}{\pi }  \lim_{\epsilon\to 0^+}\int_{-\infty}^{\lambda}
\int_{0}^{\infty }{\rm Im }\left(\Phi(s)R(t-\epsilon i;
T)\Phi(s)\right)\,ds \,dt   \\
=&\frac{1}{{\pi}^{3/2} } \lim_{\epsilon\to 0^+}
\int_{-\infty}^{\lambda} \int_{0}^{\infty }{\rm Im }\left( i
e^{-\frac{s^2}{2}} s^{ \frac{it+\epsilon-1}{2} }\int_s^\infty w^{
\frac{-it-\epsilon-1}{2}
}e^{-\frac{w^2}{2}}\,dw \right)    \,ds \,dt   \\
=&\frac{1}{{\pi}^{3/2} } \lim_{\epsilon\to
0^+}\int_{-\infty}^{\lambda} \int_{0}^{\infty } {\rm Re
}\left(e^{-\frac{s^2}{2}} s^{ \frac{it+\epsilon-1}{2}
}\int_s^\infty w^{ \frac{-it-\epsilon-1}{2}
}e^{-\frac{w^2}{2}}\,dw\right) \,ds
\,dt\\
=& \frac{1}{{\pi}^{3/2} } \lim_{\epsilon\to
0^+}\int_{-\infty}^{\lambda} \int_{0}^{\infty }\int_s^\infty
e^{-\frac{s^2+w^2}{2}}  {\rm Re }\left(s^{ \frac{it+\epsilon-1}{2}
}w^{ \frac{-it-\epsilon-1}{2} }\right)\,dw \,ds \,dt \ .
\end{align*}
 Since, for $s>0$,
\begin{align*}
 s^{ \frac{it+\epsilon-1}{2} }&=s^{ \frac{\epsilon-1}{2} }\left( \cos\left(\frac{t}{2} \ln s  \right)
 +i \sin\left( \frac{t}{2} \ln s \right)
 \right)\  ,\\
 w^{ \frac{-it-\epsilon-1}{2} }&=w^{ -\frac{\epsilon+1}{2} }\left( \cos\left(\frac{t}{2} \ln w  \right)
 -i \sin\left( \frac{t}{2} \ln w \right)  \right)\  ,
\end{align*}
we see that
\[
{\rm Re }\left(s^{ \frac{it+\epsilon-1}{2} } w^{
\frac{-it-\epsilon-1}{2} }\right)=s^{ \frac{\epsilon-1}{2} } w^{-
\frac{\epsilon+1}{2} }\cos\left(\frac{t}{2} \ln \frac{s}{w}
\right) \ .
\]
Mathematica computes
\begin{align*}
 \int_{0}^{\infty }\int_s^\infty
e^{-\frac{s^2+w^2}{2}}s^{ \frac{\epsilon-1}{2} } w^{-
\frac{\epsilon+1}{2} }   &  \cos\left(\frac{t}{2} \ln \frac{s}{w}
\right)\,dw \,ds =\left(
\frac{\pi}{2}\right)^{1/2}\\
&\cdot\left(\frac{1}{1-it+\epsilon}\,\,
{}_{2}F_1\left(\frac{1}{2}, \frac{1-it+\epsilon}{4};
\frac{5-it+\epsilon}{4};-1\right)\right.\\
&\,\,\,\,\,\left.+\frac{1}{1+it+\epsilon}\,\,
{}_{2}F_1\left(\frac{1}{2}, \frac{1+it+\epsilon}{4};
\frac{5+it+\epsilon}{4};-1\right)\right) \  ,
\end{align*}
so
\begin{align*}
&\lim_{\epsilon\to 0^+} \int_{0}^{\infty }\int_s^\infty
e^{-\frac{s^2+w^2}{2}}s^{ \frac{\epsilon-1}{2} } w^{-
\frac{\epsilon+1}{2} }   \cos\left(\frac{t}{2} \ln \frac{s}{w}
\right)\,dw \,ds =\left( \frac{\pi}{2}\right)^{1/2}\\
&\cdot\left(\frac{1}{1-it}\,\, {}_{2}F_1\left(\frac{1}{2},
\frac{1-it}{4}; \frac{5-it}{4};-1\right)+\frac{1}{1+it}\,\,
{}_{2}F_1\left(\frac{1}{2}, \frac{1+it}{4};
\frac{5+it}{4};-1\right)\right) \ ,
\end{align*}
where
\[
{}_{2}F_1\left(a, b;  c; z\right)=\sum_{n=0}^\infty \frac{
(a)_n(b)_n }{(c)_n }\frac{z^n}{n!}\ ,
\]
is \textit{Gauss's  hypergeometric function}. Using the identity
\[
{}_{2}F_1\left(a, b;  a+1; z\right)=\frac{a}{z^a} B(z; a, 1-b)\ ,
\]
we obtain
\begin{align*}
\lim_{\epsilon\to 0^+}& \int_{0}^{\infty }\int_s^\infty
e^{-\frac{s^2+w^2}{2}}s^{ \frac{\epsilon-1}{2} } w^{-
\frac{\epsilon+1}{2} } \cos\left(\frac{t}{2} \ln \frac{s}{w}
\right)\,dw \,ds=\left(\frac{\pi}{2}\right)^{1/2}\frac{1}{4}\\
&\cdot \left((-1)^{-\frac{1-it}{4}}B\left(-1;
\frac{1}{4}\left(1-it\right), \frac{1}{2}\right)+
(-1)^{-\frac{1+it}{4}}\,\,B\left(-1; \frac{1+it}{4},
\frac{1}{2}\right) \right)\  ,
\end{align*}
which, since
\[
(-1)^{-\frac{1-it}{4}}=\frac{\sqrt{2}}{2}(1-i)e^{-\frac{t\pi}{4} }
\,,\,(-1)^{-\frac{1+it}{4}}=\frac{\sqrt{2}}{2}(1-i)e^{\frac{t\pi}{4}
}\ ,
\]
gives
\begin{align*}
 \lim_{\epsilon\to 0^+} \int_{0}^{\infty }\int_s^\infty &
e^{-\frac{s^2+w^2}{2}}s^{ \frac{\epsilon-1}{2} } w^{-
\frac{\epsilon+1}{2} } \cos\left(\frac{t}{2} \ln \frac{s}{w}
\right)\,dw \,ds= \frac{{\pi}^{1/2}(1-i)}{8}\\
&\cdot\left( e^{-\frac{t\pi}{4} } \,\, B\left(-1;
\frac{1}{4}\left(1-it\right), \frac{1}{2}\right)+
e^{\frac{t\pi}{4} }\,\,B\left(-1; \frac{1+it}{4},
\frac{1}{2}\right) \right)\  .
\end{align*}
Thus
\begin{align*}
\langle \Phi, E_\lambda \Phi\rangle &=\frac{1}{{\pi}^{3/2} }
\int_{-\infty}^{\lambda} \frac{{\pi}^{1/2}(1-i)}{8}\\
&\cdot\left( e^{-\frac{t\pi}{4} } \,\, B\left(-1;
\frac{1}{4}\left(1-it\right), \frac{1}{2}\right)+
e^{\frac{t\pi}{4} }\,\,B\left(-1; \frac{1+it}{4},
\frac{1}{2}\right) \right)    \,dt\\
&=\frac{1-i}{8\pi}\,\int_{-\infty}^\lambda \left(e^{-\frac{\pi
t}{4}}\,B\left(-1; \frac{1- i t}{4},\frac{1}{2}
\right)+e^{\frac{\pi t}{4}}\,B\left(-1; \frac{1+ i
t}{4},\frac{1}{2} \right)\right)\,dt \  .
\end{align*}
Therefore, for the vacuum characteristic function we have
\begin{align*}
\langle \Phi, e^{it(XP+PX)}\Phi
\rangle&=\int_{\mathbb{R}}e^{it\lambda}\,d\langle\Phi,
E_{\lambda}\Phi \rangle\\
=\frac{1-i}{8\pi}\,\int_{\mathbb{R}}e^{it\lambda}&\,\left(e^{-\frac{\pi
\lambda}{4}}\,B\left(-1; \frac{1- i \lambda}{4},\frac{1}{2}
\right)+e^{\frac{\pi \lambda}{4}}\,B\left(-1; \frac{1+ i
\lambda}{4},\frac{1}{2} \right)\right)\,d \lambda\ ,
\end{align*}
which, using the series representation of the incomplete Beta
function, yields
\begin{align*}
\langle \Phi, e^{it(XP+PX)}\Phi
\rangle&=\frac{2}{\pi^{1/2}}\,\sum_{n=0}^\infty
\frac{\left(\frac{1}{2}\right)_n (-1)^n (4n+1) }{n! }\\
&\cdot\left((2
\pi)^{-1/2}\,\int_{\mathbb{R}}\frac{e^{it\lambda}}{(4n+1)^2+\lambda^2}\,d
\lambda\right)\ .
\end{align*}
From Fourier transform tables, wee see that
\[
(2
\pi)^{-1/2}\,\int_{\mathbb{R}}\frac{e^{it\lambda}}{(4n+1)^2+\lambda^2}\,d
\lambda=\frac{1}{4n+1} \left(\frac{\pi}{2}\right)^{1/2}
e^{-(4n+1)|t|}.
\]
Thus
\begin{align*}
\langle \Phi, e^{it(XP+PX)}\Phi
\rangle&=\frac{1}{\sqrt{2}}\,\sum_{n=0}^\infty
\frac{\left(\frac{1}{2}\right)_n (-1)^n e^{-(4n+1)|t|} }{n! }\\
&=\frac{1}{\sqrt{2}}\,\frac{e^{-|t|}}{1+e^{-4|t|}}
=\sqrt{\frac{2}{e^{2|t|}+e^{-2|t|}}}\\
&=\sqrt{\frac{2}{e^{2t}+e^{-2t}}}=\sqrt{{\rm sech} \, 2t} \  .
\end{align*}
\end{proof}

\begin{remark} \rm
As pointed out in \cite{BFspec}, using
\[
X=\frac{a+a^{\dagger}}{\sqrt{2}}
\,\,,\,\,P=\frac{a-a^{\dagger}}{\sqrt{2}i}  \  ,
\]
where
\[
\lbrack a, a^{\dagger}\rbrack=\mathbf{1},
\]
we find that
\[
XP+PX=i\left( (a^{\dagger})^2-a^2  \right) \ ,
\]
so the result of   Theorem \ref{XP+PX} is in agreement with
Proposition 3.4 of \cite{AccBouCOSA} (see also Proposition 3.9 of
\cite{AccBouArxiv},  Proposition 4.1.1 of \cite{Fein} and
Proposition 4 of \cite{AccBouPhys}), and Proposition 9.5 of
\cite{BFspec}, where it is shown that $XP+PX$ is a continuous
binomial or Beta process. We remark also that, in the statement of
Proposition 9.5 of \cite{BFspec}, a factor of $2$ is missing.
\end{remark}

\section{The quantum harmonic oscillator Hamiltonian operator
$\frac{1}{2}(X^2+P^2)$}

\begin{theorem}\label{sa2} The operator $\frac{1}{2}(X^2+P^2)$ is symmetric on the Schwartz class $\mathcal{S}$
and admits a self-adjoint extension.
\end{theorem}

\begin{proof} Let $T=\frac{1}{2}(X^2+P^2)$ and let $f,g
\in \mathcal{S}$. As in Theorem \ref{sa1}, since $Xf, Pf \in
\mathcal{S}$, $T$ is well defined on ${\rm dom}(T)=\mathcal{S}$.
Since $X, P$ are self-adjoint,
\begin{align*}
\langle Tf, g\rangle=&\frac{1}{2} \langle X^2f, g
\rangle+\frac{1}{2} \langle P^2f, g \rangle =\frac{1}{2}\langle
Xf, Xg \rangle+ \frac{1}{2}\langle Pf, Pg
\rangle \\
=&\frac{1}{2}\langle f, X^2g \rangle+\frac{1}{2}\langle f, P^2g
\rangle =\frac{1}{2}\langle f, \left(X^2+P^2\right)g \rangle
=\langle f, Tg \rangle\ ,
\end{align*}
so $T$ is symmetric. To show that $T$ admits a self-adjoint
extension, we notice that if $f$ is in the domain of $T$ then so
also is its complex conjugate $\bar{f}$. Moreover, if $f(x)$ is a
real function in the domain of  $T$ then
$Tf(x)=\frac{1}{2}\left(x^2 f(x)-f^{\prime\prime}(x) \right)$ is
also real. Thus, by a theorem due to von Neumann (see
\cite{Yosida}, Chapter XI, Section 7,  Theorem 1), $T$ admits a
self-adjoint extension.

\end{proof}

\begin{theorem}\label{RR1}The resolvent of
 $\frac{1}{2}\left(X^2+P^2\right)$ is
\begin{align*}
R&\left(a; \frac{1}{2}\left(X^2+P^2\right)\right)
g(s)=c_1(a)\,M_1(s\sqrt{2}; a)+c_2(a) \, M_2(s\sqrt{2}; a)\\
&+\int_{-\infty}^{s\sqrt{2}} g\left(\frac{w}{\sqrt{2}}\right)
\left( M_1(w; a) M_2(s\sqrt{2}; a)-M_1(s\sqrt{2}; a)M_2(w; a)
\right)\,dw\ .
\end{align*}
where $g\in \mathcal{S}$,  $c_1(a),c_2(a)\in \mathbb{C} $,
\[
M_1 (z; a)=e^{-\frac{z^2}{4}}
{}_1F_1\left(-\frac{a}{2}+\frac{1}{4}; \frac{1}{2}; \frac{z^2}{2}
\right)\,,\,M_2(z; a)= z e^{-\frac{z^2}{4}}
{}_1F_1\left(-\frac{a}{2}+\frac{3}{4}; \frac{3}{2}; \frac{z^2}{2}
\right)\ ,
\]
and
\[
{}_{1}F_1\left(a;  c; z\right)=\sum_{n=0}^\infty \frac{ (a)_n
}{(c)_n }\frac{z^n}{n!} \  ,
\]
is  \textit{Kummer's confluent hypergeometric function}.
\end{theorem}

\begin{proof}

We notice that  for $s\in\mathbb{R}$,   with all derivatives taken
with respect to $s$,
\begin{align*}
R\left(a;\frac{1}{2}\left(X^2+P^2\right)\right)g(s)=G(s; a)& \iff g(s)=\left(a-\frac{1}{2}\left(X^2+P^2\right)\right)G(s; a)\\
& \iff g(s)=\left(a-\frac{1}{2}s^2\right) G(s; a)+\frac{1}{2} G^{\prime\prime} (s; a)    \\
& \iff G^{\prime\prime} (s; a)-\left(s^2-2a\right) G(s; a)= 2 g(s)
\ ,
\end{align*}
i.e. $G$ satisfies a \textit{non-homogeneous Weber differential
equation} \cite{web1,  par, zj}.

Letting $M(s; a):=G\left(\frac{s}{\sqrt{2}}; a  \right)$, we see
that $M$ satisfies  a non-homogeneous Weber differential equation
in \textit{canonical form} \cite{web1,  par},
\[
M^{\prime\prime} (s; a)-\left(\frac{1}{4}s^2-a\right) M(s; a)=
g\left(\frac{s}{\sqrt{2}}\right) \  .
\]
The general solution of the associated homogeneous differential
equation
\[
M^{\prime\prime} (s; a)-\left(\frac{1}{4}s^2-a\right) M(s; a)=0 \
,
\]
is, see \cite{zj} and \cite{web1},
\[
M (s; a)=c_1(a)\,M_1(s; a)+c_2(a) \, M_2(s; a) \  ,
\]
where, $c_1(a),c_2(a)\in \mathbb{C} $ and
\[
M_1 (s; a)=e^{-\frac{s^2}{4}}
{}_1F_1\left(-\frac{a}{2}+\frac{1}{4}; \frac{1}{2}; \frac{s^2}{2}
\right)\,,\,M_2(s; a)= s e^{-\frac{s^2}{4}}
{}_1F_1\left(-\frac{a}{2}+\frac{3}{4}; \frac{3}{2}; \frac{s^2}{2}
\right) \  .
\]
The Wronskian $W(M_1, M_2)(s)$ of $M_1$ and $M_2$ is identically
equal to $1$ (an easy way to show this is by showing that
$\frac{dW}{ds}(M_1, M_2)(s; a)=0$ and then computing $W(M_1,
M_2)(0; a)=1$). By the well-known \textit{variation of parameters
formula}, a  solution of the non-homogeneous Weber differential
equation is
\[
M_p(s; a)=\int_{-\infty}^s g\left(\frac{w}{\sqrt{2}}\right) \left(
M_1(w; a) M_2(s; a)-M_1(s; a) M_2(w; a) \right)\,dw \  .
\]
Thus, the general solution of the non-homogeneous Weber
differential equation, in canonical form, is
\begin{align*}
M(s; a)&=c_1(a)\,M_1(s; a)+c_2(a) \, M_2(s; a)\\
&+\int_{-\infty}^s g\left(\frac{w}{\sqrt{2}}\right) \left( M_1(w;
a) M_2(s; a)-M_1(s; a)M_2(w; a) \right)\,dw\  .
\end{align*}
 Thus,
\begin{align*}
G(s; a)&=M(s\sqrt{2}; a)=c_1(a)\,M_1(s\sqrt{2}; a)+c_2(a) \, M_2(s\sqrt{2}; a)\\
&+\int_{-\infty}^{s\sqrt{2}} g\left(\frac{w}{\sqrt{2}}\right)
\left( M_1(w; a) M_2(s\sqrt{2}; a)-M_1(s\sqrt{2}; a)M_2(w; a)
\right)\,dw\ ,
\end{align*}
where,
\begin{align*}
M_1 (s\sqrt{2}; a)&=e^{-\frac{s^2}{2}}
{}_1F_1\left(-\frac{a}{2}+\frac{1}{4}; \frac{1}{2}; s^2
\right)\  ,\\
M_2(s\sqrt{2}; a)&= s\sqrt{2} e^{-\frac{s^2}{2}}
{}_1F_1\left(-\frac{a}{2}+\frac{3}{4}; \frac{3}{2}; s^2 \right) \
.
\end{align*}
\end{proof}

\begin{theorem}\label{webercf} The vacuum spectral resolution of
  $\frac{1}{2}\left(X^2+P^2\right)$ is
\[
\langle \Phi, E_\lambda \Phi\rangle=H_{1/2}(\lambda) \ .
\]
 Moreover, for $t\in
\mathbb{R}$,
\[
\langle \Phi, e^{i t \frac{1}{2}\left(X^2+P^2\right)} \Phi
\rangle=e^{\frac{it}{2}} \ ,
\]
i.e., the probability distribution of
 $\frac{1}{2}\left(X^2+P^2\right)$ is degenerate with pdf $\delta_\frac{1}{2}$ .
\end{theorem}
\begin{proof}  For $\Phi(s)=\pi^{-1/4}\,e^{-\frac{s^2}{2}}$,
\begin{align*}
\langle \Phi, & E_\lambda \Phi\rangle =\frac{1}{\pi }
\lim_{\epsilon\to 0^+}\int_{-\infty}^{\lambda}{\rm Im }\langle
\Phi,R\left(t-\epsilon i;
\frac{1}{2}\left(X^2+P^2\right)\right)\Phi\rangle\,dt\\
&=\frac{1}{\pi }  \lim_{\epsilon\to
0^+}\int_{-\infty}^{\lambda}{\rm Im }\left( \int_{-\infty}^{\infty
}\Phi(s)R\left(t-\epsilon i;
\frac{1}{2}\left(X^2+P^2\right)\right)\Phi(s)\,ds \right)\,dt \ .
\end{align*}
For $a=t-i\epsilon$, in the notation of Theorem \ref{RR1}, we
have,
\[
G(s;
t-i\epsilon)=R\left(t-i\epsilon;\frac{1}{2}\left(X^2+P^2\right)\right)\Phi(s)
\ ,
\]
where
\[
G^{\prime\prime} (s; t-i\epsilon)-\left(s^2-2(t-i\epsilon)\right)
G(s; t-i\epsilon)=2 \pi^{-1/4}\,e^{-\frac{s^2}{2}} \  .
\]
Letting
\begin{align*}
G(s; t-i\epsilon)&=K(s; t-i\epsilon)+iL(s; t-i\epsilon)\ ,\\
K(s; t-i\epsilon)&={\rm Re }\,G(s; t-i\epsilon)\ , \\
L(s; t-i\epsilon)&={\rm Im }\,G(s; t-i\epsilon)\ ,
\end{align*}
we find that the real-valued functions $K, L$ satisfy the system
of ODE's:
\begin{align*}
K^{\prime\prime} (s; t-i\epsilon)-\left(s^2-2t\right) K(s;
t-i\epsilon)+2\epsilon L &=
2\pi^{-1/4}\,e^{-\frac{s^2}{2}} \  ,\\
L^{\prime\prime} (s; t-i\epsilon)-\left(s^2-2t\right) L(s;
t-i\epsilon)-2 \epsilon K&= 0 \ .
\end{align*}
For $\epsilon\to 0$ the system uncouples into,
\begin{align*}
K^{\prime\prime} (s; t)-\left(s^2-2 t\right) K(s; t)&=
2\pi^{-1/4}\,e^{-\frac{s^2}{2}} \  ,\\
L^{\prime\prime} (s; t)-\left(s^2-2 t\right) L(s; t)&= 0 \ .
\end{align*}
 The equation satisfied by $L$ is a \textit{homogeneous Weber
differential equation} \cite{web1, zj}. Thus,  in the notation of
Theorem \ref{RR1},
\[
L(s; t)=c_1(t)\,M_1(s\sqrt{2}; t)+c_2(t) \, M_2(s\sqrt{2}; t) \ ,
\]
where, $c_1(t), c_2(t)\in\mathbb{R}$. Therefore,
\[
\langle \Phi,  E_\lambda \Phi\rangle =\frac{1}{\pi }
\int_{-\infty}^{\lambda}\int_{-\infty}^{\infty } \Phi(s) L(s;
t)\,ds \,dt \  .
\]
Using,
\[
\int_{\mathbb{R}}s^{2n}e^{-s^2}\,
ds=\Gamma\left(n+\frac{1}{2}\right)\,,\,\int_{\mathbb{R}}s^{2n+1}e^{-s^2}\,
ds=0\,,\,n\in\{0, 1, 2,...\}
 \ ,
\]
we obtain
\[
\langle \Phi,  E_\lambda \Phi\rangle =\frac{1}{\pi^{3/4} }
\int_{-\infty}^{\lambda}c_1(t) \sum_{n=0}^\infty \frac{
\left(\frac{1}{4}-\frac{t}{2} \right)_n }{n!} \,dt \ .
\]
To determine $c_1(t)$ we will use the fact that
\[
\langle \Phi,  E_\infty \Phi\rangle =1 \  .
\]
The partial sums of the series
\[
 \sum_{n=0}^\infty \frac{
\left(x \right)_n }{n!}
\]
are
\[
s_k=\sum_{n=0}^k \frac{ \left(x \right)_n }{n!}=\frac{(1+k)
\Gamma(1+x+k)}{ \Gamma (x+1) \Gamma (2+k)} \ .
\]
Thus,
\[
\sum_{n=0}^\infty \frac{ \left(x \right)_n
}{n!}=\lim_{k\to\infty}s_k=\left\{
\begin{array}{llr}
  0\   ,&  x< 0   \\
& \\
1\ , & x=0\\
& \\
 \infty \   ,&  x> 0
 \end{array}
 \right.
\  .
\]
Thus, in order for
\[
\langle \Phi,  E_\infty \Phi\rangle =1 \  ,
\]
we must interpret $c_1(t)$ in the distribution sense as,
\[
c_1(t)=\pi^{3/4} \delta_{1/2}\left(t \right)\  ,
\]
in which case,
\[
\langle \Phi,  E_\lambda \Phi\rangle =
\int_{-\infty}^{\lambda}\delta_{1/2}\left(t\right)
\sum_{n=0}^\infty \frac{ \left(\frac{1}{4}-\frac{t}{2} \right)_n
}{n!} \,dt =\left\{
\begin{array}{llr}
  1\   ,&  \lambda \geq \frac{1}{2}   \\
& \\
0\   ,&  \lambda < \frac{1}{2}
 \end{array}
 \right\}
=H_{1/2}(\lambda)\ .
\]

Thus,
\begin{align*}
 \langle \Phi, e^{i t \frac{1}{2}\left(X^2+P^2\right)} \Phi
\rangle&=\int_{\mathbb{R}}e^{it\lambda}\,d\langle\Phi,
E_{\lambda}\Phi
\rangle\\
&=\int_{\mathbb{R}}e^{it\lambda}\,dH_{1/2}(\lambda)=\int_{\mathbb{R}}e^{it\lambda}\,\delta_{1/2}\left(\lambda
\right)\,d\lambda=e^{\frac{it}{2}}\  ,
 \end{align*}
 meaning that the probability distribution of
 $\frac{1}{2}\left(X^2+P^2\right)$ is degenerate.
\end{proof}
\begin{remark} \rm In terms of the creation and annihilation operators
$a^{\dagger}$ and $a$, respectively, where
\[
\lbrack a, a^{\dagger}\rbrack=\mathbf{1} \  ,
\]
using
\[
X=\frac{a+a^{\dagger}}{\sqrt{2}}
\,\,,\,\,P=\frac{a-a^{\dagger}}{\sqrt{2}i} \  ,
\]
we find that
\[
\frac{1}{2}\left(X^2+P^2\right)=\frac{1}{2}+a^{\dagger}a \ .
\]
The formula for the characteristic function given in Theorem
\ref{webercf} is precisely a special case of Remark 3.7 of
Proposition 3.4 of \cite{AccBouCOSA}, where
$\frac{1}{2}+a^{\dagger}a$ is denoted by $S^1_1$.
\end{remark}

\section{Stone's Formula in the finite dimensional case}

To illustrate the use of Stone's formula in the finite dimensional
case, we consider the \textit{Heisenberg observable}
\[
H=\begin{pmatrix}0&1&1\\
1&0&1\\
1&1&0
\end{pmatrix} \ ,
\]
of Section 6 of \cite{BFspec}. The eigenvalues of $H$ are:
$\lambda_1=2$, with multiplicity one,  and $\lambda_2=-1$ with
multiplicity two. It was shown in  \cite{BFspec} that, if
$\Phi=(a,b, c)$ is a unit vector in $\mathbb{R}^3$ then, for
$t\in\mathbb{R}$:
\[
\langle \Phi, e^{itH}\Phi
\rangle=\left(1-\frac{(a+b+c)^2}{3}\right) e^{-i
t}+\frac{(a+b+c)^2}{3}e^{2 i t},
\]
i.e., $H$ follows a Bernoulli distribution with probability
density function
\begin{align*}
p_{a,b,c}(\lambda)&=\left(1-\frac{(a+b+c)^2}{3}\right)
\delta_{-1}( \lambda)+\frac{(a+b+c)^2}{3}\delta_2( \lambda).
\end{align*}

We will show how the same result can be obtained with the use of
Stone's formula.

For $a\in\mathbb{C}\setminus\{-1, 2\}$ we have
\[
R\left(a; H\right)=\frac{1}{(a+1)(a-2)}\begin{pmatrix} a-1&1&1\\
1&a-1&1\\
1&1&a-1
\end{pmatrix} \ ,
\]
and
\begin{align*}
R(t-i\epsilon; &H)-R(t+i\epsilon; H)=\frac{2
i \epsilon}{((t-2)^2+\epsilon^2)((t+1)^2+\epsilon^2  )}\\
&\cdot\begin{pmatrix}
 3+t(t-2)+\epsilon^2&2t-1&2t-1\\
2t-1&3+t(t-2)+\epsilon^2&2t-1\\
2t-1&2t-1&3+t(t-2)+\epsilon^2
\end{pmatrix} \ .
\end{align*}

For a unit vector $\Phi=(a,b, c)$  in $\mathbb{R}^3$, we find,
\begin{align*}
\langle \Phi, &\left(R(t-i\epsilon; H)-R(t+i\epsilon;
H)\right)\Phi \rangle=\frac{2
i \epsilon}{((t-2)^2+\epsilon^2)((t+1)^2+\epsilon^2  )}\\
&\cdot \left( (t-2)^2+\epsilon^2+(2t-1)(a+b+c)^2 \right)\ ,
\end{align*}
and
\begin{align*}
\frac{1}{2\pi i }\int \langle \Phi,&(R(t-i\epsilon;
H)-R(t+i\epsilon; H))\Phi \rangle \,dt\\
&=\frac{1}{3\pi}\left( x^2 \arctan \left( \frac{t-2}{\epsilon}
\right)-(x^2-3)\arctan \left( \frac{t+1}{\epsilon} \right)
\right)\ ,
\end{align*}
where,
\[
x=a+b+c \ .
\]
Since,
\[
\lim_{y\to -\infty}  \frac{1}{3\pi}\left( x^2 \arctan \left(
\frac{y-2}{\epsilon} \right)-(x^2-3)\arctan \left(
\frac{y+1}{\epsilon} \right) \right) =-\frac{1}{2} \ ,
\]
we find that
\begin{align*}
\frac{1}{2\pi i }\int_{-\infty}^\lambda \langle
\Phi,&(R(t-i\epsilon;
H)-R(t+i\epsilon; H))\Phi \rangle\,dt\\
&=\frac{1}{3\pi}\left( x^2 \arctan \left(
\frac{\lambda-2}{\epsilon} \right)-(x^2-3)\arctan \left(
\frac{\lambda+1}{\epsilon} \right) \right)+\frac{1}{2}\ .
\end{align*}
Since,
\begin{align*}
\lim_{\epsilon\to 0^+} &\frac{1}{3\pi}\left( x^2 \arctan \left(
\frac{\lambda-2}{\epsilon} \right)-(x^2-3)\arctan \left(
\frac{\lambda+1}{\epsilon} \right) \right)\\
&=\left\{
\begin{array}{llr}
  \frac{1}{2}\   ,&  \lambda >2  \\
& \\
-\frac{x^2}{3}+\frac{1}{2}\   ,& -1< \lambda < 2 \\
& \\
-\frac{1}{2}\   ,& \lambda < -1
 \end{array}
 \right.
\ ,
\end{align*}
we have,
\[
\langle \Phi,  E_\lambda \Phi\rangle =\left\{
\begin{array}{llr}
  1\   ,&  \lambda >2  \\
& \\
-\frac{x^2}{3}+1\   ,& -1< \lambda < 2 \\
& \\
0\   ,& \lambda < -1
 \end{array}
 \right.
\ .
\]
Using the right-continuity of the spectral resolution, we extend
to $\lambda=2$ and $\lambda=-1$, to obtain the vacuum resolution
of the identity,
\[
\langle \Phi,  E_\lambda \Phi\rangle =\left\{
\begin{array}{llr}
 1\   ,&  \lambda \geq 2  \\
& \\
-\frac{x^2}{3}+1\   ,& -1\leq \lambda < 2 \\
& \\
0\   ,& \lambda < -1
 \end{array}
 \right\}
=\left(-\frac{x^2}{3}+1\right)H_{-1}(\lambda)+\frac{x^2}{3}
H_2(\lambda) \ .
\]
Therefore,
\[
d\,\langle \Phi,  E_\lambda \Phi\rangle =
\left(\left(-\frac{x^2}{3}+1\right)\delta_{-1}(\lambda)+\frac{x^2}{3}
\delta_2(\lambda)\right)\,d\lambda \ ,
\]
and
\begin{align*}
\langle \Phi, e^{itH}\Phi
\rangle&=\left(-\frac{x^2}{3}+1\right)\int_\mathbb{R}e^{it\lambda}\delta_{-1}(\lambda)\,d\lambda
 +\frac{x^2}{3}\int_\mathbb{R}e^{it\lambda}
 \delta_2(\lambda)\,d\lambda\\
&=\left(1-\frac{x^2}{3}\right)e^{-it}
 +\frac{x^2}{3}e^{2it} \  ,
\end{align*}
in agreement with the result obtained in Section 6 of
\cite{BFspec}.

\section{Errata}

The following corrections should me made to \cite{BFspec}:

\medskip

Formulas (8.5) and (8.7) are correct for positive $a, b$. For
general $a, b$, with $ab\neq 0$, they should be replaced,
respectively, by
\begin{align}
q(s)&= \frac{a}{2\sqrt{ab} }\tanh (2\sqrt{ab}\,s),   \\
r(s)&=\frac{b}{2\sqrt{ab} }\tanh (2\sqrt{ab}\,s) \  .
\end{align}

The formula in Theorem 8.5 is
\[
\langle \Phi, e^{it (a X^2+b P^2)}\Phi \rangle=
\frac{\sqrt{2}\,e^{\frac{1}{2}p(it)}}{\sqrt{e^{2 p(it)}+(2
\,q(it)-1)(2 \,r(it)-1)}}.
\]
and, as a result, the formula in Corollary 8.6 is
\[
\langle \Phi, e^{itH} \Phi \rangle= e^{\frac{it}{2}}
  \left( \frac{{\rm sech }\left(\frac{t\sqrt{3}}{2}\right)
}{1-\frac{i}{\sqrt{3}} \tanh \left(\frac{t\sqrt{3}}{2}\right)}
\right)^{1/2}
 \ .
\]

\end{document}